\documentclass[aip,jmp,reprint,amsmath,amssymb]{revtex4-1}

\usepackage{color}
\usepackage{booktabs}
\usepackage{amsthm}
\newtheorem{lemma}{Lemma}
\newtheorem{theorem}{Theorem}

\usepackage[normalem]{ulem}
\usepackage{epsfig}
\usepackage{natbib}
\usepackage{comment}
\usepackage{amsmath,amssymb,graphicx,url,gensymb,amsbsy}
\usepackage{tikz}
\usepackage{graphicx}
\usepackage[caption=false]{subfig}
\usepackage{amsmath,amssymb,graphicx,url,gensymb,amsbsy}
\usetikzlibrary{calc, 
		arrows,
                petri,
                topaths} 
\usepackage{tkz-berge}
\def\eqalign#1{\null\,\vcenter{\openup\jot\m@th 
   \ialign{\strut\hfil$\displaystyle{##}$  &  $\displaystyle{{}##}$\hfil \crcr#1\crcr}}\,} 
\newcommand{\ra}[1]{\renewcommand{\arraystretch}{#1}}

\newcommand\floor[1]{\lfloor#1\rfloor}

\begin{document}

\title{On extreme points of the diffusion polytope}

\author{M. J. Hay}
\email[]{hay@princeton.edu}
\affiliation{Department of Astrophysical Sciences, Princeton University, Princeton, New Jersey 08544}

\author{J. Schiff}
\affiliation{Department of Mathematics, Bar-Ilan University, Ramat Gan, 52900 Israel}

\author{N. J. Fisch}
\affiliation{Department of Astrophysical Sciences, Princeton University, Princeton, New Jersey 08544}
\affiliation{Princeton Plasma Physics Laboratory, Princeton, New Jersey 08543}

\date{\today}

\begin{abstract}
We consider a class of diffusion problems defined on simple graphs in which the populations at any two 
vertices may be averaged if they are connected by an edge. 
The diffusion polytope is the convex hull of the set of population vectors attainable using finite sequences 
of these operations. 
A number of physical problems have linear programming solutions taking the diffusion polytope 
as the feasible region, e.g. the free energy that can be removed from plasma using waves, 
so there is a need to describe and enumerate its extreme points. 
We review known results for the case of the complete graph $K_n$, and study a variety of problems for 
the path graph $P_n$ and the cyclic graph $C_n$. We describe the different kinds of extreme points that arise,
and identify the diffusion polytope in a number of simple cases. In the case of 
increasing initial populations on $P_n$  the diffusion polytope is topologically 
an $n$-dimensional hypercube. 
\end{abstract}


\maketitle


\section{Introduction}

Consider a discrete-time conservative diffusion process on a graph. By this we mean  a connected, 
simple graph $G$ with 
vertices $\{V_i\}_{i=1}^n$, a set of initial ``populations'' $\{\rho_i\}_{i=1}^n$ at the vertices, 
and a set of rules that can be applied at each time step, with the understanding that the rules 
in some sense diffuse, or spread out the populations, while conserving the total $\sum_i \rho_i$. So, for 
example, in classical diffusion,\cite{ref1,kondor} at  each time step the populations at all the 
vertices are updated simultaneously via the rule
\begin{gather*}
\rho_i \rightarrow  \rho_i + h \sum_{j} (\rho_j - \rho_i)\, \qquad i=1,\ldots,n\ , 
\end{gather*}
where $h$ is a positive constant and the index $j$ runs over the neighbors of vertex $i$.

Chip firing games on graphs are specified in a similar fashion. At each time step the (integer) population at a vertex is reduced by $n$, the degree of vertex $i$, and the population at each of $i$'s neighbors is incremented by $1$.\cite{ref2} Sandpile models, in which vertices accumulate population until reaching a threshold and `toppling,' thereby transferring population to other nodes, follow similar rules.\cite{ref3a,ref3b}

In this paper we will examine a diffusion process in which at each time step, 
the populations at any two vertices connected by an edge can be averaged. 
In some generality, for many reasonable sets
of rules, there will a bounded set of attainable population vectors in ${\bf R}^n$. 
In various applications we may be 
interested in extremizing some linear function 
of the populations $\sum_i  w_i\rho_i$, and to do this 
(using a linear programming approach) we need to 
identify the closure of the convex hull of the set of attainable population vectors. 
We call this the {\em diffusion polytope} of the diffusion process (associated with the graph
$G$, the relevant sent of rules  and the initial set of populations). 
In some sense the diffusion polytope measures the diversity 
of behavior that can be attained in the diffusion process. 
(We emphasize that the diffusion processes we consider in this paper are limited to those described, as above, by a set of rules, leading to a finite, or at least bounded, set of accessible states. This does not include typical stochastic diffusion processes, in which extreme states are in principle accessible, albeit with very small probabilities.)

Our motivation comes from plasma physics.
There is a class of diffusion problems assoicated with opportunities in extracting energy in plasma with waves. 
Waves can be injected into a fusion reactor such that high energy alpha particles, the byproducts of the fusion reaction, 
lose energy to the waves, as those alpha particles are diffused by the waves to lower energy.\cite{fisch92,13a,tutorial}
The extra energy in the waves can then be used, for example, to increase the reactivity of the fuel or to drive electric 
current. \cite{fisch94,hay15a,natreview}
Choosing the correct sequence of waves to extract as much energy as possible is an optimization problem 
on a graph of the type described above. Because the wave-particle interaction is diffusive,\cite{sagdeev} particles in a given 
location in the 6D phase space of velocity and position are randomly mixed by the wave with particles in another location 
in the 6D phase space. A graph encodes the connectedness of the phase space; each node denotes a volume element
in phase space, and an edge between two nodes indicates that those two volume elements may be mixed.
Wave-induced diffusion follows a path in the 6D phase space corresponding to an edge in the associated graph.
If there is a population inversion in energy along the path, then the diffusive process releases energy.  
If there is only one wave and a specified diffusion path, the amount of extractable energy can be readily calculated. 
However, more energy can be extracted when several waves are employed.\cite{fisch_twowave,herrmann97}
When many diffusion paths are possible, it turns out that the order in which these paths are taken affects the 
energy that can be extracted. We reach the situation described above, of a range of possible 
population distributions on the nodes. 
(The number of particles ­­­ the total population ­­­ is conserved, energy is extracted by moving particles from nodes of high energy to nodes of low energy.)
Determining the maximum amount of extractable 
energy under the constraint that the particle distribution function evolves only due to diffusion reduces to 
a linear programming problem on the diffusion polytope, the convex hull of all attainable 
population vectors, and we are concerned with identifying its extreme 
points and the edge sequences that give rise to these extreme points. 

Arguably the simplest case of this diffusion problem is to allow, at every step, 
averaging  of the populations at any two nodes.   
This is the full-connectivity, or the {\em nonlocal  diffusion} problem, in which diffusion paths can be constructed between 
any two phase space locations, and the relevant graph is the complete graph $K_n$ 
(see Figure~\ref{fig:graphs}).  In the context of plasma, 
there are physical reasons why this arrangement is realizable  on a macroscopic 
scale, despite the restriction of diffusion to 
contiguous regions of phase space on the microscopic scale. \cite{fisch93}
In this case the diffusion polytope, and the maximum energy extractable, or what we call the {\it free energy}, 
have been described previously.\cite{hay15} The same optimization problem has been discussed in other fields, 
in the context of attainable states in chemical reactions and thermal processes\cite{horn64,zylka85}, and 
in the context of altruism and wealth distribution.\cite{thonwallace}

However in all these settings, there are arguments to restrict the connectivity. In the context of plasma,
possible reasons for restriction include that waves can only diffuse particles from one phase space 
position to a contiguous position,
or between pairs of states determined by selection rules. 
In the case of such a {\it local diffusion} problem, the free energy will be less, because there are fewer ways 
in which the energy might be released. 
As the simplest example of such a local diffusion problem we study diffusion when the connectivity is restricted 
from that of $K_n$ to that  of the path graph $P_n$ (see Figure~\ref{fig:graphs}). 
In the context of alruism, 
the effects of other (retrospective) restrictions on the connectivity have also been studied.\cite{aboudithon}

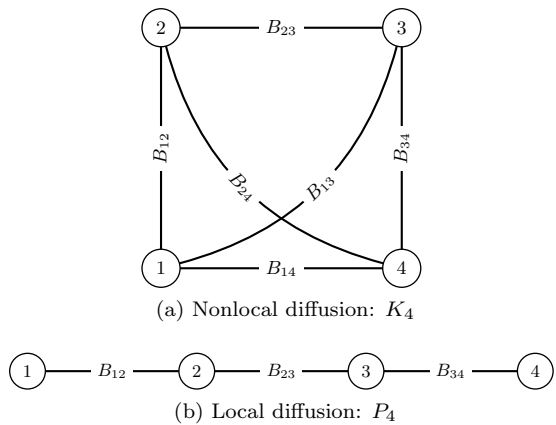
\begin{figure}

\subfloat[Nonlocal diffusion: $K_4$]{%
\begin{tikzpicture}[scale=.8,transform shape]
  \Vertex[x=0,y=0]{1}
  \Vertex[x=0,y=4]{2}
  \Vertex[x=4,y=4]{3}
  \Vertex[x=4,y=0]{4}
  \tikzstyle{LabelStyle}=[fill=white,sloped]
  \Edge[label=$B_{12}$](1)(2)
  \Edge[label=$B_{23}$](2)(3)
  \Edge[label=$B_{34}$](3)(4)
  \Edge[label=$B_{14}$](1)(4)
\tikzstyle{EdgeStyle}=[bend right]
  \Edge[label=$B_{13}$](1)(3)
  \Edge[label=$B_{24}$](2)(4)
\end{tikzpicture}
}

\subfloat[Local diffusion: $P_4$]{%
\begin{tikzpicture}[scale=0.75,transform shape]
  \Vertex[x=0,y=0]{1}
  \Vertex[x=3,y=0]{2}
  \Vertex[x=6,y=0]{3}
  \Vertex[x=9,y=0]{4}
  \tikzstyle{LabelStyle}=[fill=white,sloped]
  \Edge[label=$B_{12}$](1)(2)
  \Edge[label=$B_{23}$](2)(3)
  \Edge[label=$B_{34}$](3)(4)
\end{tikzpicture}
}
\caption{Graph representations diffusion problems two four-level systems, $K_4$ and $P_4$. 
The marking  $B_{ij}$ on an edge indicates that the populations of nodes $i$ and $j$ can be 
equalized. \label{fig:graphs}}

\end{figure}

The path graph context $P_n$ arises naturally when the 6D phase space is projected to a 1D energy representation, and only 
transitions between adjacent energies are allowed. 
Many physical problems of interest are captured by the model of  contiguity based on energy only.
Other network problems can be defined which capture the spatial element.\cite{zhmoginov08} 
A different problem on $P_5$ arises in the context of 
maximizing  the possible concentration of atoms in a specific state in an  ensemble of atoms of helium.
Fig.~\ref{fig:helium} shows a truncated level diagram for parahelium ($S=0$). Due to level splitting, each 
energy level is associated with a unique energy. We suppose  processes are available  which can mix levels joined 
by a dipole transition (such as spatially incoherent light at the appropriate frequency).  
For example, it is possible to average the number of atoms in $2s$ and $3p$, but it is not possible to do the same 
for $2s$ and $3s$. Thus four operators are allowed in this five-level system, $1s\leftrightarrow 2p$, 
$1s\leftrightarrow 3p$, $2s\leftrightarrow 3p$, $3s\leftrightarrow 2p$.  We gain a clearer picture by 
redrawing the energy level diagram of Fig.~\ref{fig:helium} with vertices relabelled 
$1=1s, 2=2s, 3=2p, 4=3s, 5=3p$ (so the vertex label indicates the energy rank) to obtain Fig.~\ref{fig:heliumgraph}. 
Thus we see this also gives a diffusion problem on $P_n$; however the allowed transitions are {\em not} between adjacent 
energies. 

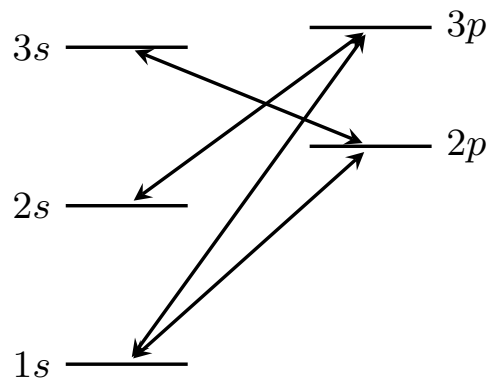
\begin{figure}
\centerline{
  \resizebox{7cm}{!}{
    \begin{tikzpicture}[
      scale=0.5,
      level/.style={thick},
      virtual/.style={thick,densely dashed},
      trans/.style={thick,<->,shorten >=2pt,shorten <=2pt,>=stealth},
      classical/.style={thin,double,<->,shorten >=4pt,shorten <=4pt,>=stealth}
    ]
    \draw[level] (2cm,-10em) -- (0cm,-10em) node[left] {$1s$};
    \draw[level] (2cm,-2em) -- (0cm,-2em) node[left] {$2s$};
    \draw[level] (4cm,1em) -- (6cm,1em) node[right] {$2p$};
    \draw[level] (2cm,6em) -- (0cm,6em) node[left] {$3s$};
    \draw[level] (4cm,7em) -- (6cm,7em) node[right] {$3p$};
    \draw[trans] (1cm,6em) -- (5cm,1em) node[midway,left] { };
    \draw[trans] (1cm,-10em) -- (5cm,1em) node[midway,left] { };
    \draw[trans] (1cm,-10em) -- (5cm,7em) node[midway,left] { };
    \draw[trans] (1cm,-2em) -- (5cm,7em) node[midway,left] { };
    \end{tikzpicture}
  }
}
\caption{Truncated parahelium energy level diagram with all possible electric dipole transitions
 represented by arrows, $\leftrightarrow$. Energy scale is arbitrary. The $s$-$p$ splitting 
is due chiefly to the partially screened Coulomb repulsion of the nucleus (a larger effect in $p$ 
orbitals vs. $s$ orbitals).\cite{schwabl}\label{fig:helium}}
\end{figure}

\begin{figure}
 \begin{tikzpicture}[scale=1,transform shape]
  \Vertex[x=0,y=0]{1}
  \Vertex[x=-2,y=-2]{3}
  \Vertex[x=-2,y=-4]{4}
  \Vertex[x=2,y=-2]{5}
  \Vertex[x=2,y=-4]{2}
  \tikzstyle{LabelStyle}=[fill=white,sloped]
  \Edge[label=$B_{13}$](1)(3)
  \Edge[label=$B_{15}$](1)(5)
  \Edge[label=$B_{34}$](4)(3)
  \Edge[label=$B_{25}$](5)(2)
\end{tikzpicture}
\caption{Graph representation for the diffusion problem on parahelium, cf. Fig.~\ref{fig:helium}. 
Vertices are labeled by the rank of the corresponding energy eigenvalue (increasing). The 
graph is $P_5$, but the transitions are not between adjacent energy levels.\label{fig:heliumgraph}}
\end{figure}

This paper proceeds as follows: 
In Section \ref{sec:model}  we give the precise statement of the diffusion
model we study, the definition of the diffusion polytope, and state some elementary facts.  
In Section \ref{sec:oldresults}  we review known results for the 
case of the complete graph $K_n$,\cite{hay15,horn64,zylka85,thonwallace}. We show that for a general graph, 
whenever $3$ vertices are connected to each other (i.e. form a triangle), extreme points of the diffusion 
polytope are obtained only by averaging over the pairs with consecutive populations. This generalizes a 
theorem of Thon and Wallace for case of the complete graph, and is a key result in characterizing the 
diffusion polytope. 
In Sec.~\ref{sec:localstate}, we study the diffusion polytope for the path graph $P_n$. 
There are different cases depending on the ordering of the initial populations. In the case $n=3$ we describe 
the solution in all cases, emphasizing the location of the resulting polytopes inside the $K_3$ polytope,
and the different kind of extreme points that arise.  In the case of $P_n$ with ordered initial 
populations we show the diffusion polytope is topologically 
an $(n-1)$-dimensional hypercube with $2^{n-1}$ vertices. 
Whereas the extreme points of the $K_n$ nonlocal problem can all be constructed by ${n\choose 2}$ or 
fewer level mixings, some extreme points in the $P_n$ local problem are only reachable 
by an infinite sequence of operations. Curiously, the number extreme points in the $P_n$ local problem 
that are inherited from the nonlocal problem is a Fibonacci number, and the number of operations 
required to reach them is at most $\floor{n/2}$.   
Sec.~\ref{sec:cycle} extends the analysis to diffusion on the cycle graph $C_n$, again focusing on the 
new types of extreme points that become available. 
We summarize in Sec.~\ref{sec:discussion}, 
and present open questions 
concerning  more physically relevant graphs, 
and connections between ideas presented in this paper and other notions in modern network theory. 


\section{The Diffusion Model\label{sec:model}}

The diffusion model studied in this paper is as follows: We are given a connected, 
simple graph $G$ with vertices $\{V_i\}_{i=1}^n$, 
and a set of initial populations $\{\rho_i\}_{i=1}^n$ at the vertices. 
We assume without loss of generality that the population vector $\rho=(\rho_1,\rho_2,\ldots,\rho_n)$ 
is normalized: $\sum_i \rho_i = 1$. 

To any edge in $G$ we associate an operator. If the edge connects vertices $V_i$ and $V_j$ we indicate
this operator $B_{ij}$, and this acts on the populations at the vertices $i$ and $j$ via
$$  B_{ij} : (\rho_i,\rho_j) \rightarrow \left( \frac12(\rho_i+\rho_j), \frac12(\rho_i+\rho_j) \right) $$ 
while leaving the populations at all the other vertices unchanged. 

We can write $B_{ij} = \frac12 \left(I + Q_{ij} \right)$ where $Q_{ij}$ is the operator that permutes
the populations at the $i$'th and $j$'the vertices. In greater 
generality we could consider the action of operators 
$B_{ij;\alpha} = (1-\alpha) I + \alpha Q_{ij}$  for all $\alpha\in[0,\frac12]$. This is the case in which 
``partial relaxation'' is allowed as well as full relaxation. However, since we will only consider 
the convex hull of the population vectors generated by the $B_{ij}$, it is clear that this does not make 
any difference. However, it is important to distinguish between the cases of $0\leq\alpha\leq1/2$ and 
$1/2<\alpha\leq1$; the latter case corresponds to inversion of populations, and are not allowed. 

We assume we are given an objective function $f=\sum_i  w_i\rho_i$ which is to be extremized over the set 
$A({ \rho}_0)$ of attainable states, i.e. 
populations generated by finite sequences  of the operators $B_{ij}$ from the initial population vector ${ \rho}_0$. 
The weights $w_i$ are taken to be all positive 
and distinct. Without loss of generality we can assume either $w_1<w_2<\cdots<w_n$ or that the components of 
${\rho}_0$ 
satisfy $\rho_1\leq\rho_2\leq\cdots\leq\rho_n$. Due to the linearity of $f$, this problem has a linear 
programming solution on $DP=\overline{ch\left(A({\rho}_0) \right)}$, the closure of the complex hull of 
$K({ \rho}_0)$. We call this the {\em diffusion polytope} of the problem; it is determined by the graph $G$ and the 
initial population vector ${ \rho}_0$.  

Since we have assumed $G$ is connected, the uniform population  $\rho = \left( \frac1{n}, \frac1{n}, \ldots, \frac1{n} \right)$ 
is in $DP$. (To prove this, observe that the quantity ${\rm max}_{i,j}(\rho_i-\rho_j)$ is a strictly decreasing function under
the application of the averaging operations, and it cannot have a non-zero minimum.) 

Another immediate property of $DP$ is that if the graph $G'$ can be obtained from $G$ by deletion of one or more 
edges (while still staying connected) then $DP(G')\subseteq DP(G)$. Here the assumption is that we start with the
same population vector on both $G$ and $G'$. However, since $G'$ has less edges, the set of attainable states is smaller 
compared to that for $G$ (and in the plasma setting the free energy is reduced). Thus the diffusion polytope of 
{\em every} graph with $n$ vertices is a subset of the diffusion polytope for $K_n$. Reducing the connectivity will 
restrict the diffusion polytope. This gives, for example, a way to identify ``important'' edges in a graph,
as edges whose elimination causes a significant restriction diffusion polytope, or to define robustness of 
a network in terms of how the diffusion polytope responds to removal of edged, c.f. Refs.~\onlinecite{ref4a,ref4b}.


\section{Diffusion on $K_n$: Nonlocal Diffusion\label{sec:oldresults}}

For the case of $K_n$, the case of nonlocal diffusion, Ref.~\onlinecite{thonwallace} presented a recursive algorithm to identify 
the extreme points of the diffusion polytope by applying different sequences of the $B_{ij}$. It was noted in 
Ref.~\onlinecite{hay15} that this algorithm generates reduced (minimal length) decompositions of the elements in the 
symmetric group $S_n$.\cite{oeis_reduced} That is, the algorithm identifies every possible minimum-length way to generate 
each of the $n!$ permutations of a length-$n$ word using only adjacent transpositions $\sigma_i=(i\,\,i+1)$. 
It turns out that the  nonlocal extreme points are in bijection with equivalence classes of reduced decomposition, 
the equivalence classes being the sets  of reduced decompositions obtainable from each other by the 
applying the commutation  relation $\sigma_i\sigma_j=\sigma_j\sigma_i$, $|i-j|>1$.\cite{berkolaiko16}

It was noted that `dead ends' exist among the possible sequences of diffusion operations, where a state is 
reached with level densities decreasing with level energy, such that no more energy can be extracted. Any such 
{\it stopping} state has the level population permutation which is the reverse of the energy level permutation. 
For example, given $w=(w_1,w_2,w_3)$ with $w_2<w_1<w_3\sim\{2,1,3\}$, the {\it stopping permutation} is $\{3,1,2\}$, 
such that $\rho_3\leq \rho_1\leq \rho_2$. 

In order to identify the extremal sequence of diffusion operations resulting in a minimum-energy state, 
it is generally necessary to evaluate the objective function for each inequivalent reduced decomposition of 
the stopping permutation. For the worst-case reverse permutation, there are a large number of states to 
check.\cite{oeis_networks} We emphasize that each extreme point in the nonlocal problem is associated 
with a finite sequence of diffusion operations, with the limiting words being the identity (length 0) 
and the reverse permutation (length ${n\choose 2}$).

Finally, all extremal sequences result in a monotone trend in the objective function: in our plasma 
example, we can exclude from consideration any operations which absorb energy from the injected waves.


A signficant tool for these results was Prop.2 in Ref.~\onlinecite{thonwallace} and this has an
extension to the diffusion problem on an arbitrary graph. Suppose the $3$ vertices $V_i,V_j,V_k$ form 
a triangle, i.e. that there are edges between the $3$ possible pairs of these three vertices. Assume 
without loss of generality that $\rho_i<\rho_j<\rho_k$. Then {\em no extreme point can be obtained by 
immediate application of $B_{ik}$}. In other words, in any triangle, extreme points can only 
generated by averaging pairs with adjacent populations. This fact folows from  two following simple identities:
\begin{eqnarray*}
&& \left( \begin{array}{ccc}
\frac12 & 0 & \frac12 \\ 
0 & 1 & 0  \\
\frac12 & 0 & \frac12 
\end{array} \right) 
\left( \begin{array}{c} 
a \\ b \\ c 
\end{array} \right)  \\
&=& 
\lambda_1
\left( \begin{array}{c} 
\frac13 \\ \frac13 \\\frac13 
\end{array} \right) 
+ 
(1-\lambda_1)   
\left( \begin{array}{ccc}
\frac12 & 0 & \frac12 \\ 
0 & 1 & 0  \\
\frac12 & 0 & \frac12 
\end{array} \right) 
\left( \begin{array}{ccc}
1 & 0 & 0 \\ 
0 & \frac12 & \frac12  \\
0 & \frac12 & \frac12 
\end{array} \right) 
\left( \begin{array}{c} 
a \\ b \\ c 
\end{array} \right)  \\
&=& 
\lambda_2
\left( \begin{array}{c} 
\frac13 \\ \frac13 \\\frac13 
\end{array} \right) 
+ 
(1-\lambda_2)   
\left( \begin{array}{ccc}
\frac12 & 0 & \frac12 \\ 
0 & 1 & 0  \\
\frac12 & 0 & \frac12 
\end{array} \right) 
\left( \begin{array}{ccc}
\frac12 & \frac12  & 0 \\
\frac12 & \frac12  & 0 \\
0 & 0 & 1 
\end{array} \right) 
\left( \begin{array}{c} 
a \\ b \\ c 
\end{array} \right) 
\end{eqnarray*}
Here it is assumed that $a<b<c$ and $a+b+c=1$, and
$$ 
\lambda_1 =  \frac{3(c-b)}{b+c-2a} \quad {\rm and} \quad 
\lambda_2 =  \frac{3(b-a)}{2c-a-b} \ ,
$$ 
so
$$ 
1-\lambda_1 =  \frac{2(2b-a-c)}{b+c-2a} \quad {\rm and} \quad 
1-\lambda_2 =  \frac{2(a+c-2b)}{2c-b-a} \ . 
$$ 
If $a+c\le 2b$ then $0\le \lambda_1 \le 1$,  and the first identity says that the 
population obtained by application of $B_{ik}$ is a convex combination of 
the population obtained by averaging all three populations at $V_i,V_j,V_k$, with
the population obtained by application of first $B_{jk}$ and then $B_{ik}$. The latter two 
populations may or may not be extreme points of $DP$; but they are both in $DP$, 
and thus the population obtained by immediate application of $B_{ik}$ is certainly
{\em not} an extreme point. If $a+c>2b$ then $0\le \lambda_2\le 1$ and the 
second identity is relevant, and the population obtaiined by immediate application of
$B_{ik}$ is a convex combination of the population obtained by full averaging, with
that obtained by first applying $B_{ij}$ and then $B_{ik}$. 

In contrast, the other crucial result for understanding the case of $K_n$, viz. 
Prop.3 in Ref.~\onlinecite{thonwallace}, seems to be specific to $K_n$.


\section{Diffusion on $P_n$: Local Diffusion\label{sec:localstate}}

Throughout our discussion of the $P_n$ case we assume without loss of generality that 
$\rho_1\leq\rho_2\leq\cdots\leq\rho_n$. 

For $n=3$ there are three distinct cases to consider, when
the allowed operators are (a) $B_{12},B_{23}$, (b) $B_{12},B_{13}$, (c) $B_{13},B_{23}$. 
Figure~\ref{fig:hulls} compares the diffusion polytopes in the three different cases with that
of $K_3$, in the case $\rho_0=(0,2/7,5/7)$. 

\begin{figure}

\subfloat[$B_{12}$, $B_{23}$ allowed ($P_3$)]{%
\includegraphics[width=40mm]{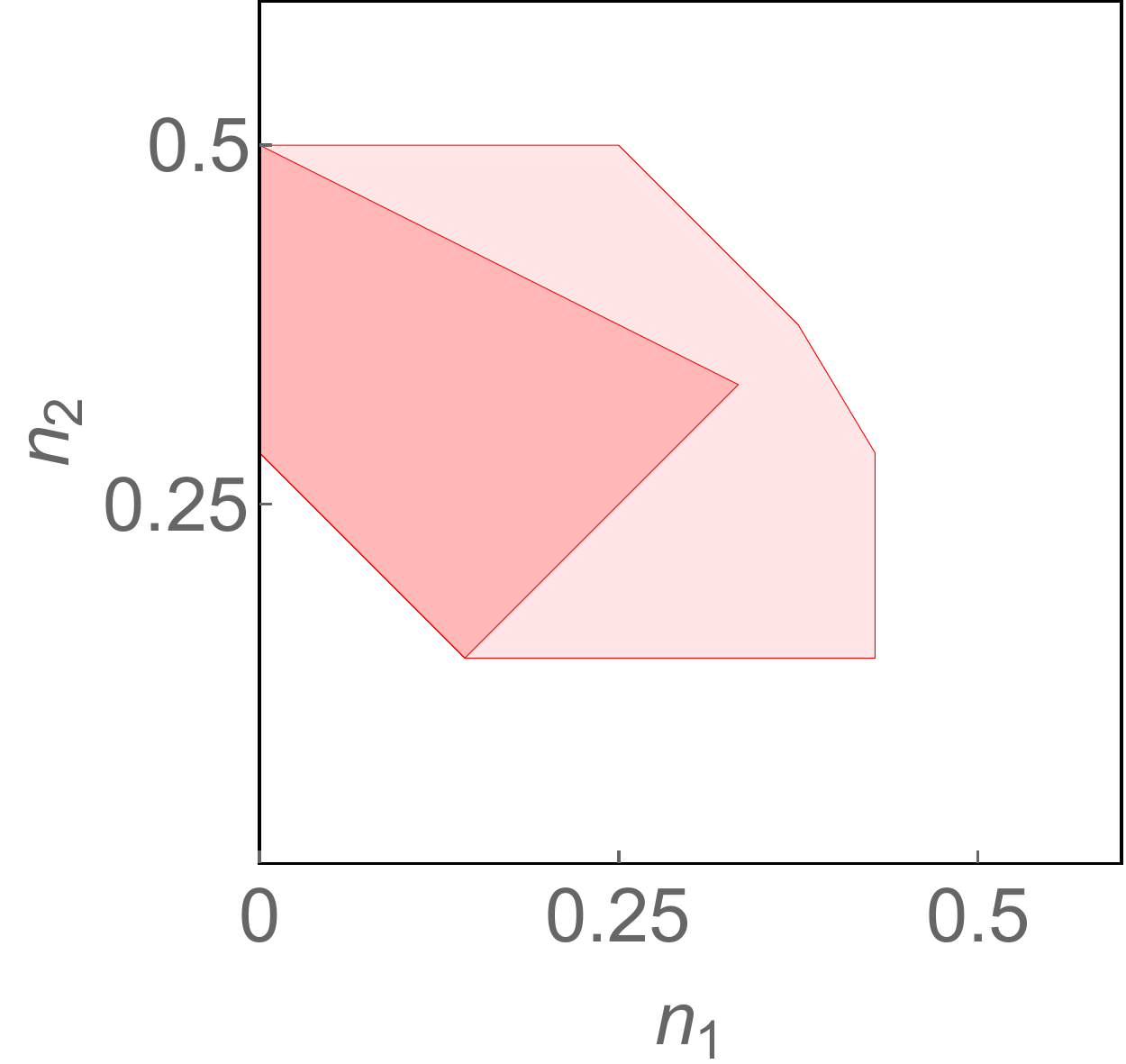}
}

\subfloat[$B_{12}$, $B_{13}$ allowed]{%
\includegraphics[width=40mm]{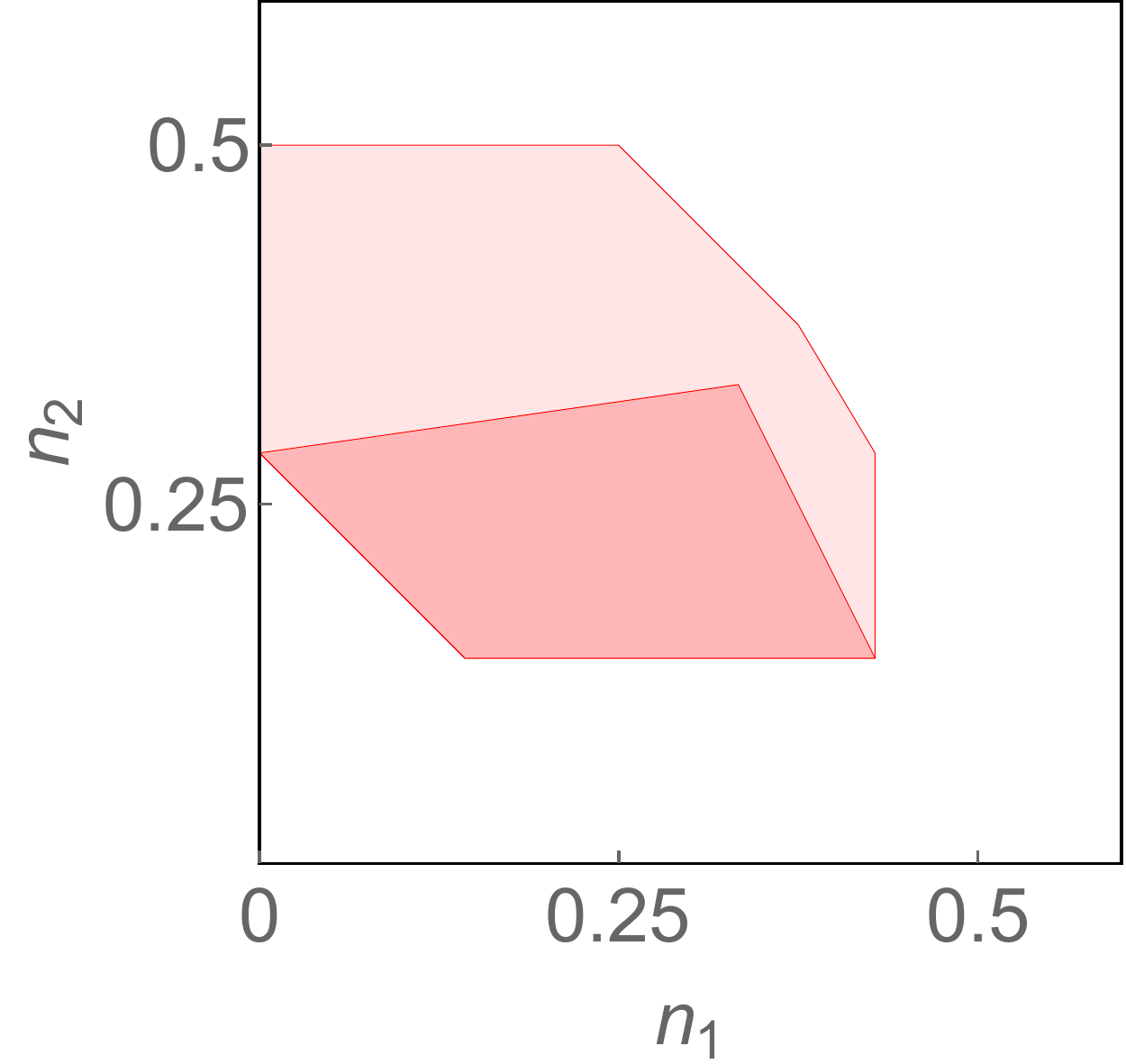}
}

\subfloat[$B_{13}$, $B_{23}$ allowed]{%
\includegraphics[width=40mm]{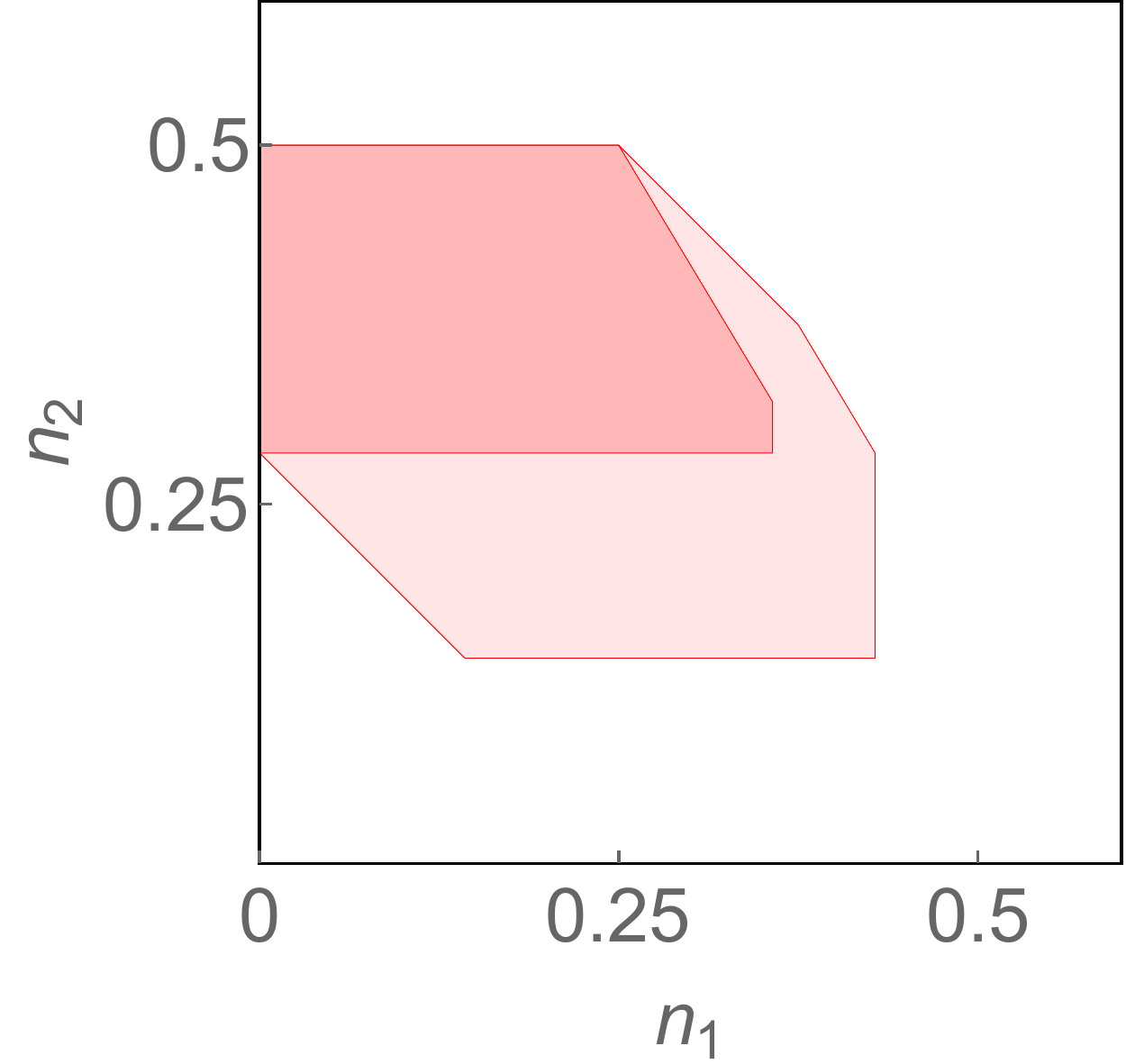}
}

\caption{Comparison of polytopes for $K_3$, $P_3$, and two other restricted graphs for initial data $\rho_0=(0,2/7,5/7)$,
with permitted operators (a) $B_{12},B_{23}$, (b) $B_{12},B_{13}$, (c) $B_{13},B_{23}$. All are superimposed on the polytope for $K_3$.
(Due to normalization, the third coordinate is ignorable.) 
 \label{fig:hulls}}

\end{figure}

The extreme points in the case $K_3$ are 
$\rho_0,\   \rho_0 B_{12},\  \rho_0 B_{23},\   \rho_0 B_{12}B_{13},\ \rho_0 B_{23}B_{13},\ 
\rho_0 B_{12}B_{13}B_{23},$  $\rho_0 B_{23}B_{13}B_{12}$. 
Writing $\bar{\rho}=\left( \frac13 ,\frac13 ,\frac13 \right)$, 
the extreme points in the three $P_3$ cases are 
\begin{itemize}  
\item $\rho_0,  \rho_0 B_{12}, \rho_0 B_{23}, \bar{\rho}$. 
\item $\rho_0,  \rho_0 B_{12}, \rho_0 B_{12}B_{13}, \bar{\rho}$. 
\item $\rho_0,  \rho_0 B_{13}, \rho_0 B_{23}, \rho_0 B_{13}B_{23}, \rho_0 B_{23}B_{13} $. 
\end{itemize}  
Note first that any $K_3$ extreme point that can be attained in any of 
the $P_3$ cases is an extreme point in that case. We call such points ``nonlocal
extreme points'' (of the local problem), as they are inherited form the 
nonlocal problem. However, there are also new extreme points. 
In the first two cases the point $\bar{\rho}$ is added. This is a limit point of the 
attainable populations --- it cannot be attained by application of a finite sequence of 
the $B_{ij}$ operators; we call such points ``asymptotic extreme points''.  
In the third case there are new extreme points that {\em can} 
be attained by application of a finite sequence of the $B_{ij}$, however these are not 
extreme points for the nonlocal $K_3$ case. In more general local diffusion problems all 
three kinds of extreme points coexist --- nonlocal extreme points inherited from $K_n$, 
asymptotic extreme points (which do not appear in the case $K_n$) and other extreme 
points generated by finite sequences of operations that are {\em not} inherited from
$K_n$. 

In the case $P_n$ in which the allowed operators are $B_{12},B_{23},\ldots,B_{n-1,n}$ the 
diffusion polytope can be identified explicitly. In the appendix, we prove that there  
are $2^{n-1}$ extreme points in bijection with the power 
set of $\{1,2,\ldots,n-1\}$. 
It follows that the diffusion polytope is topologically an $(n-1)$-dimensional hypercube.\cite{mathworldhyper}
Any extreme point corresponding to a subset $A\subseteq \{1,2,\ldots,n-1\}$ is connected 
(by edges, forming the 1-skeleton of the hypercube) to $n-1$ other extreme points corresponding
to the $n-1$ subsets that differ from $A$ in just one element. 
Fig.~\ref{fig:3d} illustrate the  3-cube hull for the four-level problem. 

\begin{figure}
\centering
\includegraphics[width=70mm]{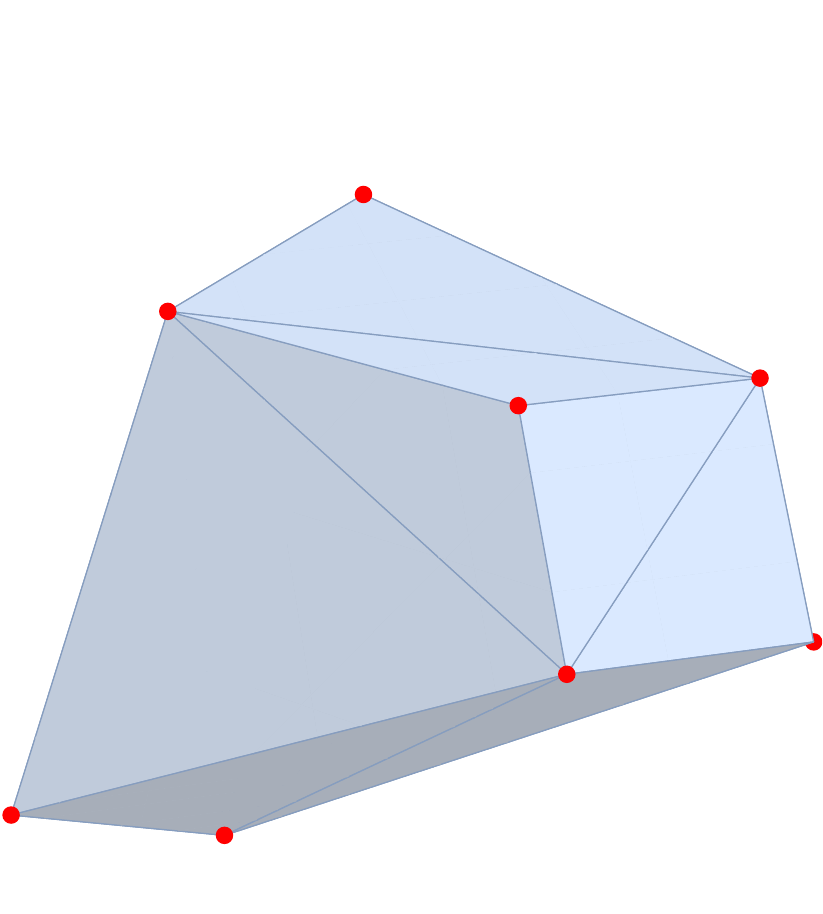}
\caption{\label{fig:3d} Convex hull of a four-level local diffusion problem represented in $\mathbb{R}^3.$ 
Extreme points are denoted with red circles.}
\end{figure}

We ask the question how many nonlocal extreme points are there in this case? (i.e. how many points are 
inherited from the case of the complete graph $K_n$.)  In the case $n=3$ there are $3$: 
$\rho_0$, $\rho_0B_{12}$, and $\rho_0B_{23}$ are all extreme points. In the case $n=4$ there are
$5$: $\rho_0$, $\rho_0B_{12}$, $\rho_0B_{23}$, $\rho_0B_{34}$ and $\rho_0B_{12}B_{34}$. In general, the
question is how many subsets of {\em commuting} operators are there in
$\{B_{12},B_{23},\ldots,B_{n-1,n}\}$?  For $n=4$ and $n=5$ only commuting 2-tuples are possible. 
For $n=6$, a commuting 3-tuple appears: $(B_{12},\,B_{34},\,B_{56})$. In general, the $n$-level system 
contains only $k$-tuples satisfying $k\leq\floor{n/2}$.

Clearly, as $n$ grows, the number of commuting $k$-tuples with $k\leq\floor{n/2}$ 
becomes large and direct counting becomes tedious, if not difficult. Fortunately, a general formula is 
available. Appropriating the notation of Ref.~\onlinecite{erdos76}, denote the number of 
commuting $k$-tuples as $A_k(n)$. Recalling that there are $n-1$ operators $B_{i,i+1}$ in the $n$-level problem, 
the number of extreme points for $n>2$ levels is
\begin{gather}
1+(n-1)+A_2(n)+A_3(n)+\cdots +A_{\floor{n/2}}(n),
\end{gather}
where the leading $1$ corresponds to the initial distribution $\rho_0$.

We can now attack the $A_k(n)$ in turn. Mapping each $B_{i,i+1}$ to the symbol $i$, $A_2(n)$ is the number 
of two-element subsets of $\{1,2,\ldots,n-1\}$ which do not contain consecutive numbers. The total number 
of two-element subsets is ${n-1\choose 2}$ and the number of subsets containing consecutive numbers is $n-2$. Therefore 
\begin{gather}
A_2(n)={n-1\choose 2}-(n-2)= \frac12 (n-1)(n-2) =  T_{n-3}
\end{gather}
where $T_n$ is the $n^{\rm th}$ triangular number (recall that the expression is restricted to $n>2$). 
Proceeding analogously, $A_3(n)$ is seen to correspond to the tetrahedral numbers. 
In general, $A_k(n)$ can be identified with the set of regular $k$-polytopic numbers. \cite{deza12}
With the proper offsets, the formula for the number of extreme points for $n>2$ is
\begin{gather}
n+{n-2\choose 2}+{n-3\choose 3}+\cdots=F_{n+1},
\end{gather}
where $F_n$ is the $n^{\rm th}$ term of the Fibonacci sequence: $0,\,1,\,1,\,2,\,3\ldots\,$ with $F_0 = 0$. 
The identity is the statement that shallow diagonals of Pascal's triangle sum to Fibonacci numbers:\cite{mathworldpascal}
\begin{gather}
\sum_{k=0}^{\floor{n/2}}{n-k \choose k} = F_{n+1}.
\end{gather}
(Note that ${n\choose 0}+{n-1\choose 1}=n.$)

This result might have been anticipated because there are $F_{n+2}$ unique subsets of $\{1,\ldots,n\}$ which 
do not contain consecutive numbers.\cite{comtet} For example, there are three ($=F_4$) such subsets of 
$\{1,\,2\}$: $\{\varnothing,\,\{1\},\,\{2\}\}$. The problem of enumerating the extreme points is analogous: 
the $n$-level system contains $n-1$ operators, leading to $F_{n+1}$ Fibonacci subsets.

The Fibonacci numbers are the solution to a similar problem in graph theory: $F_{n+1}$ is the number of matchings in a 
path graph with $n$ vertices.\cite{prodinger82}

Thus the number of nonlocal extreme points in the case of $P_n$ with operators $\{B_{12},B_{23},\ldots,B_{n-1,n}\}$ 
is $F_{n+1}$. Note these involve at most $\floor{n/2}$ operators (as opposed to up to ${n\choose 2}$ in the case 
of $K_n$). All the other extreme points are asymptotic extreme points, involving averagings over $3$ or more states. 



\section{Diffusion on the Cycle Graph $C_n$: Nonlocal Diffusion\label{sec:cycle}}

Another possible restriction on the phase space connectivity results in a diffusion problem on the 
cycle graph $C_n$, Fig.~\ref{fig:cyclic} with  allowed operators $\{B_{12},B_{23},\ldots,B_{n-1,n}\}$ 
as in the case of $P_n$ studied before, and  the single extra operator $B_{1,n}$ 
$C_3$ is isomorphic to $K_3$ and accordingly has the same diffusion polytopes; $C_4$ is 
the smallest case with a unique diffusion polytope. 

Introducing the notation $B_{ijk}$ for the operation of averaging over populations on the three 
vertices $i,j,k$ (involving an infinite sequence of operations), Table~\ref{table} lists the 18 extreme 
points we have found for the $C_4$ diffusion polytope by brute force computation. We divide the 
points into $3$ categories: those inherited from $K_n$, those shared with $P_4$ (by this we mean that 
these are extreme points for $P_4$ inherited from the $C_4$ case), and all others. In the 
case of $C_4$, as in the case of $P_n$ that we solved explicitly, there are only nonlocal extreme points and 
asymptotic extreme points. We emphasize that in general there are also points involving a finite 
sequence of $B_{ij}$ that are {\em not} inherited from $K_n$. 

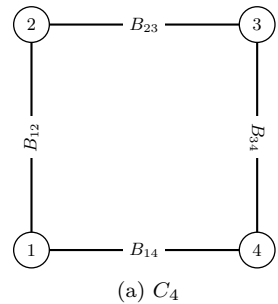
\begin{figure}
\subfloat[$C_4$]{%
 \begin{tikzpicture}[scale=0.75,transform shape]
  \Vertex[x=0,y=0]{1}
  \Vertex[x=0,y=4]{2}
  \Vertex[x=4,y=4]{3}
  \Vertex[x=4,y=0]{4}
  \tikzstyle{LabelStyle}=[fill=white,sloped]
  \Edge[label=$B_{12}$](1)(2)
  \Edge[label=$B_{23}$](2)(3)
  \Edge[label=$B_{34}$](3)(4)
  \Edge[label=$B_{14}$](1)(4)
\end{tikzpicture}
}
\caption{The cycle graph $C_4$\label{fig:cyclic}}
\end{figure}

\begin{table}\centering
\ra{1.3}
\noindent\begin{tabular}{*{3}{c}}
from $K_4$ & shared $P_4$ & $C_4$ only \\
\hline
$\rho_0$  &  $\rho_0 B_{123}$  & $\rho_0 B_{123}B_{14}$ \\
$\rho_0 B_{12}$  &  $\rho_0 B_{234}$  & $\rho_0 B_{234}B_{14}$ \\
$\rho_0 B_{23}$  &  &$\rho_0 B_{123} B_{124}$\\
$\rho_0 B_{34}$  &   &   $\rho_0 B_{234}B_{134}$\\
$\rho_0 B_{12} B_{34}$&   &$\rho_0 B_{12} B_{34} B_{134}$  \\
$\rho_0 B_{12} B_{34} B_{14}$  &   &$\rho_0 B_{12} B_{34} B_{124}$\\
 & &$\rho_0 B_{234}B_{14}B_{12}$ \\
 &  &$\rho_0 B_{123}B_{14}B_{234}$   \\
 & &$\rho_0 B_{234}B_{14}B_{123}$  \\
  &&$\rho_0 B_{123}B_{23}B_{14}B_{34}$  \\
\bottomrule
\end{tabular}
\caption{Extreme points for the diffusion problem on $C_4$.\label{table}}
\end{table}


\section{Discussion\label{sec:discussion}}

In this work we have developed a non-standard diffusion model on a graph, explained the reason for looking at the 
associated diffusion polytope, and studied this in the cases where the graph is $K_n$, $P_n$ and $C_n$. 
The case of $P_n$ for which we have given a complete solution and the case of $K_n$ should be regarded 
as extreme cases. Assuming increasing energy levels and initial densities, 
the nonlocal problem requires evaluation of the objective functional at a super-exponential number of 
points in the number of levels $n$, while for the local problem the solution is always the uniform 
distribution.  We show in Fig.~\ref{fig:local-non} the optimal population 
distribution with $w=(1,2,3,4)$, $\rho_0 \propto (e^1,e^2,e^3,e^4)$ 
in the cases of $K_4$, $P_4$ (with 
operators $B_{12},B_{23},B_{34}$) and $C_4$ (with added operator $B_{14}$). 

\begin{figure}
\centering
\includegraphics[width=70mm]{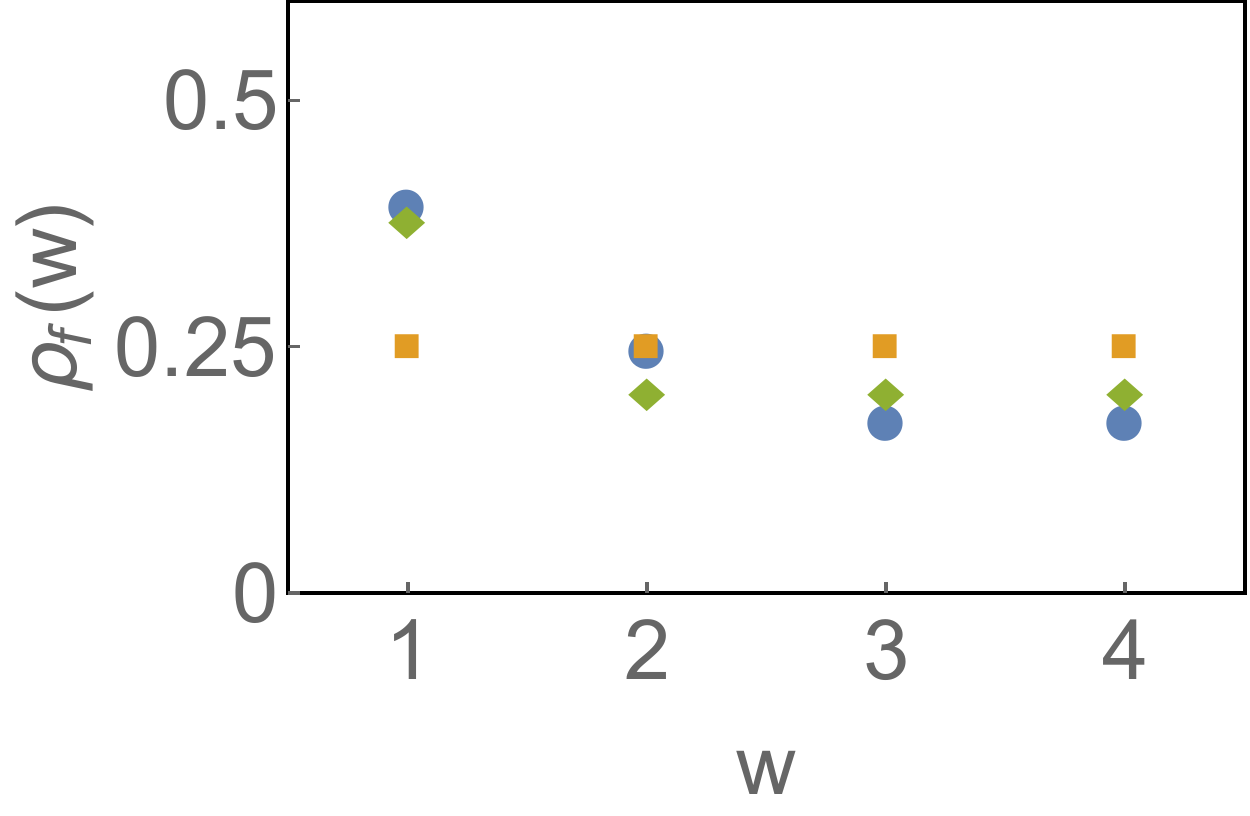}
\caption{Comparison of minimal energy states for the diffusion problem with initial data 
$\rho_0 \propto (e^1,e^2,e^3,e^4)$  and $w=(1,2,3,4)$. The blue circles label the level densities 
for the nonlocal problem ($K_4$), the green diamonds for the case $C_4$ and the yellow squares for the case 
$P_4$. Whereas in the nonlocal case the operations recovered $68\%$ of the Gardner limit,\cite{hay15} the $C_4$ case recovered $63\%$, and in the $P_4$ case  only  $50\%$ was recovered.\label{fig:local-non}}
\end{figure}

Although the  diffusion problem on the complete graph is well-characterized, it is difficult to solve.  
The $P_n$ model is an oversimplification, but it is possible that not much more complicated models may 
be useful. Some insight may be gained into the `full' alpha particle diffusion problem by considering 
proximity in both position and energy space. As before, edges correspond to diffusion paths and 
vertices reference bins associated with the discretization of configuration and energy space. 
A particularly simple 2-D diffusion problem in $x$ and $\epsilon$ can be visualized as 
in Fig.~\ref{fig:diffusionpath}, on the composition of path graphs $P_m[P_n]$, $m\geq n$, 
essentially a grid graph with diagonal edges.\cite{mathworldcomp} 
By judicious choice of the wave phase velocities, diffusion paths can be established between different parts of the phase space. Moving a particle in energy at constant location could be accomplished with a wave $k\to0$, and vice versa, moving a particle without changing its energy with a wave $\omega\to 0$. The slopes of any diagonal diffusion paths are determined by intermediate $\omega/k$.\cite{herrmann97,fisch_twowave}
Such an arrangement is characterized by degeneracy in both position and energy bins. A practical alpha 
channeling scheme would seek to maximize the density at a particular low-energy, large-radius sink node.

\begin{figure}
\begin{tikzpicture}[scale=1,transform shape]
{
    \draw [<->,thick] (0,4) node (yaxis) [above] {$x$}
        |- (6,0) node (xaxis) [right] {$\epsilon$};
\filldraw 
(1,1) circle (2pt) node[align=left,   below] {} --
(3,1) circle (2pt) node[align=center, below] {} --     
(5,1) circle (2pt) node {} --
(5,3) circle (2pt) node {}--
(3,3) circle (2pt) node {}--
(1,3) circle (2pt) node {};

\coordinate (x) at (1,3);
\coordinate (y) at (1,1);
\draw[black] (x) -- (y);

\coordinate (x) at (1,1);
\coordinate (y) at (3,3);
\draw[black] (x) -- (y);

\coordinate (x) at (1,3);
\coordinate (y) at (3,1);
\draw[black] (x) -- (y);

\coordinate (x) at (3,1);
\coordinate (y) at (5,3);
\draw[black] (x) -- (y);

\coordinate (x) at (3,3);
\coordinate (y) at (5,1);
\draw[black] (x) -- (y);

\coordinate (x) at (3,1);
\coordinate (y) at (3,3);
\draw[black] (x) -- (y);
}\end{tikzpicture}
\caption{Diffusion problem on $G = P_3[P_2]$. \label{fig:diffusionpath}}
\end{figure}
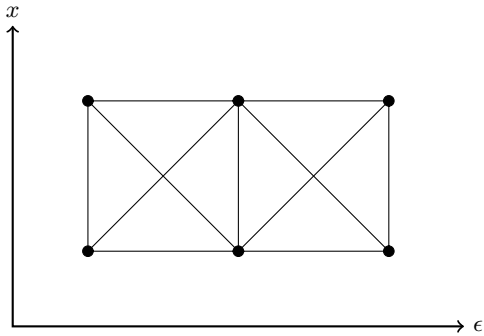

In general, one might consider diffusion problems on arbitrary subgraphs of the 
complete graph for $V=\{w\}$, of which the local problem is one solvable case (the path graph with $V=\{w\}$). 
These other problems will inherit some extreme points from the complete graph on the same number of vertices 
and should also have some unique ones depending on the particular graph structure. There may be rich 
physical significance for diffusion problems on more or less `bottlenecked' graphs generally (in the sense of Cheeger), 
or e.g. the wheel graphs, $k$-regular trees, and complete bipartite graphs.


Our problem may be contrasted with other notions of (deterministic or stochastic) diffusion or spreading on graphs, 
for example in the context of spreading of epidemics,\cite{keeling,pastor01,lloyd01} or behavior.\cite{christakis}  
In all such settings, the resulting behavior depends intimately on the character of the underlying graph. 
For example, disease transmission proceeds more slowly on a lattice than on small-world\cite{watts_strogatz,djwatts} 
or scale-free\cite{barabasi,pastor01} networks. Qualitative differences can also be observed in the context of classical 
graph diffusion.\cite{ref5} We expect to see similar differences in the context of our problem, though the 
computational task of verifying this is formidable. 


\begin{acknowledgments}
It is a pleasure to acknowledge discussions with Mariana Campos Horta. Particular thanks are due to 
Professor D. Tannor for critical discussions at an early stage of this work. Work supported by DOE 
Contract No. DE-AC02-09CH11466 and DOE NNSA SSAA Grant No. DE274-FG52-08NA28553. One of us (NJF) 
acknowledges the hospitality of the Weizmann Institute of Science, where he held a Weston Visiting 
Professorship during the time over which this work was initiated.
\end{acknowledgments}

\bibliography{bibliography}

\appendix*
\section{The diffusion polytope in the ``ordered'' case $P_n$}

In this appendix we prove the claim from Sec.~{\ref{sec:localstate}} concerning 
the diffusion polytope in the case $P_n$, with permitted operators 
$B_{12},B_{23},\ldots,B_{n-1,n}$. The local state space is the set of states accessible 
from the initial state $\rho_0$,  
after an arbitrary finite sequence of transformations $B_{i,i+1}$ has been applied. 
We claim that the closure of the convex hull of this space has $2^{n-1}$ extreme points 
in bijection with the power set of $\{1,2,\ldots,n-1\}$. Denote these points $S_A$, where the index $A$ runs 
over all subsets of $\{1,2,\ldots,n-1\}$. 

The bijection has a simple description. The empty set $\varnothing$ corresponds to the initial 
state $\rho_0$. The point $S_A$ corresponding to the set $A$ has the property that for each $i\in A$ 
the $i$'th and $(i+1)$'th components of $S_A$ are equal, and its components are obtained by averaging
over subsets of components of $\rho_0$. So, for example, 
in a seven-level system,  $S_{\{1,2,3,6\}}$ is the point
$(x,x,x,x,y,z,z)$ where $x$ is the average of the first four components of $\rho_0$, 
$y$ is the fifth component, and $z$ is the average of the sixth and seventh components. 
In greater generality, whenever $A$ contains a sequence of $k$ consecutive integers 
$i,i+1,\ldots,i+k-1$ then the $k+1$ components of $S_A$ from $i$ to $i+k$ are equal to the 
average of the corresponding components of $\rho_0$. 

\begin{lemma}
The $S_A$ are contained in the closure of the local state space.
\end{lemma}
 
\begin{proof}
Fix a subset $A\subseteq\{1,2,\ldots,n-1\}$ corresponding to a particular $S_A$. 
Suppose $i,i+1,\ldots,i+k-1$ is a maximal sequence of $k$ consecutive integers in $A$, i.e.
all these are in $A$, but $i-1$ and $i+k$ are not, and consider the quantity $\rho_{i+k} - \rho_i$. 
This is non-negative, and, if it is non-zero, strictly decreases when the operator 
$B_{i+k-1,i+k} \ldots B_{i+1,i+2} B_{i,i+1}$ is applied to the state $\rho$. Furthermore it is 
unchanged when the corresponding operator for a different maximal sequence of consecutive 
integers is applied to $\rho$. It follows that by repeated application of such operators 
the initial state $\rho_0$ can be brought arbitrarily close to the extreme point $S_A$.  
Thus the points $S_A$ are in the closure of the local state space. 
\end{proof}


\begin{lemma}
No point $S_A$ can be written as a nontrivial convex combination of the others.
\end{lemma}

\begin{proof}
Consider a convex combination $\sum a_A S_A$, with $a_A\geq 0$ and $\sum a_A=1$. Let $\rho_{A,i}$ denote 
the component densities of the point $S_A$. Suppose $\sum a_A S_A= S$, with densities $\rho_i$. Because an arbitrary 
sequence of local operators $B_{i,i+1}$ cannot achieve a population inversion, we have 
$\rho_{A,i}\leq \rho_{A,i+1}$ for all $A$ and $i$. Therefore $\rho_i=\rho_{i+1}$ if and only if $\rho_{A,i} = \rho_{A,i+1}$
for all $A$ with $a_A>0$. Thus any convex combination yielding the point $S$ has only one summand, $S$ itself.
\end{proof}

\begin{lemma}
Any point in the local state space can be written as a convex combination of the $S_A$.
\end{lemma}

\begin{proof}
We wish to show that every accessible state $\rho$ can be written as a convex combination of the 
extreme points, i.e. that we  can write 
$$  \rho = \sum_A \lambda_A S_A  \ , \qquad  {\rm where~~} \lambda_A\ge 0 \ ,  \sum_A\lambda_A = 1 . $$ 
Every accessible state can be obtained by applying an arbitrary finite sequence of $B_{i,i+1}$ to the initial state 
$\rho_0 =  S_{\varnothing}$.  We proceed by induction on the number of $B_{i,i+1}$ operators to be applied. 
To prove the inductive step it is necessary to check that for every $A$ and for every $i$,  
$S_A B_{i,i+1}$ can be written as a convex combination of extreme points. If $i \in A$, then 
$S_A B_{i,i+1}=S_A$ is an extreme point. If $i\notin A$ then the $i$'th and $(i+1)$'th components of $S_A$ 
may be different. $S_A$ then takes  the form 
$$  ( \ldots, x, \ldots, x, x, y, y, \ldots  y , \ldots )\ .   $$ 
where there is a string of $k>0$ $x$'s  ending in the $i$'th position and 
a string of $l>0$ $y$'s starting in the $(i+1)$'th position. 
The entries on the left of  the string of $x$'s and on the right of the string of $y$'s remain 
fixed and identical for all the points that appear in the forthcoming calculation, and thus
do not play any role. It is assumed that the entry immediately on the left of the $x$'s (if there is 
such) is  strictly less than $x$, and the entry immediately on the right of the $y$'s 
(if it exists) is strictly greater than $y$.  

Applying $B_{i,i+1}$ we have 
$$ S_A B_{i,i+1} = 
( \ldots, x, \ldots, x,  \frac{x+y}{2}, \frac{x+y}{2}, y, \ldots,  y , \ldots )   
$$ 
where now the strings of $x$'s and $y$'s are of length $k-1$ and $l-1$ respectively. 
It should be emphasized that except in the case $k=l=1$ {\em this is not an extreme point},
as although the values in the $i$'th and $(i+1)$'th positions are equal, they have been 
determined by components of the initial density $\rho_0$ from outside these positions, and 
full averaging over the relevant subset has not been achieved. 
However we will now show it is a convex combination of four extreme points of the following form: 
\begin{align*}
\begin{split}
 S_1 = {}&  ( \ldots,   X , \ldots X,X,X,X,\ldots  X  , \ldots ),\\ {}&{\rm with}~ k+l~X\textrm{'s,}
 \end{split}\\[0.5ex]
\begin{split} S_2 = {}&  ( \ldots, X_1, \ldots, X_1, X_1,X_1, Y_1 , \ldots,  Y_1, \ldots ),\\   {}&{\rm with} ~  k+1~X_1\textrm{'s and } l-1~Y_1\textrm{'s,}  \end{split}\\[0.5ex]
\begin{split}
 S_3 = {}&  ( \ldots, X_2, \ldots, X_2, Y_2,Y_2, Y_2 , \ldots,  Y_2, \ldots ), \\  {}&{\rm with}~ k-1~X_2\textrm{'s and } l+1~Y_2\textrm{'s, and}\end{split}\\[0.5ex]
 S_4 = {}&  ( \ldots, X_2, \ldots, X_2, Z, Z, Y_1, \ldots,  Y_1, \ldots ),  \\ {}&{\rm with}~  k-1~X_2\textrm{'s and } l-1~Y_1\textrm{'s.}
\end{align*}
(These four points are only distinct on the assumption $k,l>1$. The cases $k=l=1$, $k=1,l>1$, and $k>1,l=1$ should be 
considered separately but are simpler --- in particular in the case $k=l=1$ all four points coincide and 
the resulting point is extreme. We omit the details of the cases $k=1,l>1$ and $k>1,l=1$.)  
The quantities $x,y,X,X_1,Y_1,X_2,Y_2,Z$ appearing here are not independent, as they are obtained 
from averaging over certain entries of $\rho_0$. Denote the average value of entries $i-(k-1),\ldots,i-1$ 
of $\rho_0$ as $R_1$, the value of entry $i$ as $R_2$, the value of entry $i+1$ as $R_3$ and the average 
of entries $i+2,\ldots,i+l$ as $R_4$. Then 
\begin{eqnarray*}
x &=&  \frac{(k-1)R_1 + R_2 }{k} \\
y &=&  \frac{R_3 + (l-1)R_4 }{l} \\
X  &=&  \frac{(k-1)R_1 + R_2 + R_3 + (l-1)R_4}{k+l} \\
X_1  &=& \frac{(k-1)R_1 + R_2 + R_3}{k+1} \\
Y_1  &=&  R_4 \\
X_2 &=&  R_1  \\
Y_2 &=&  \frac{R_2 + R_3 + (l-1)R_4}{l+1} \\ 
Z  &=& \frac{R_2+R_3}{2}
\end{eqnarray*} 
Our aim is to find $\lambda_1,\lambda_2,\lambda_3,\lambda_4\ge 0$ such that 
$$ S_A B_{i,i+1} =  \sum_{i=1}^4 \lambda_i S_i ~~ {\rm and} ~~ \sum_{i=1}^4 \lambda_i = 1     $$
Solving these linear equations gives three constraints between the $\lambda_i$, which can
be written in the form 
\begin{eqnarray*} 
\lambda_1 &=& C_1 \lambda_4 + D_1  \\
\lambda_2 &=& C_2 \lambda_4 + D_2  \\
\lambda_3 &=& C_3 \lambda_4 + D_3  
\end{eqnarray*} 
where $C_1,C_2,C_3,D_1,D_2,D_3$ are complicated expressions involving $k,l$ and $R_1,R_2,R_3,R_4$. 
Explicitly we have 
\begin{align*}
\begin{split}
C_1 ={}& 
\frac { \left( p_2 +2\,p_3  \right)  \left( p_2+2\, p_1 \right)  \left( k+l \right) }
      { 2 \left( (k-1)p_1+kp_2+(k+1)p_3 \right) } \\
{}&\times\frac{1}{(l+1)p_1 +lp_2+(l-1)p_3  }\end{split} \\[0.5ex]
C_2 ={}& 
- \frac {  \left( p_2 +2\,p_3  \right)  \left( k+1 \right) }
{ 2 \left( (k-1)p_1+kp_2+(k+1)p_3 \right)  } \\[0.5ex]
C_3 ={}& 
-\frac { \left( p_2+2\,p_1 \right)  \left( l+1 \right) }
       {  2\left((l+1)p_1 +lp_2 +(l-1)p_3 \right) }
\end{align*}
where 
\begin{eqnarray*}
p_1 &=& R_2-R_1  \\
p_2 &=& R_3-R_2  \\
p_3 &=& R_4-R_3  
\end{eqnarray*}
Since $R_1\le R_2 \le R_3 \le R_4$ we have $p_i\geq 0$ and thus 
$C_1>0$ and $C_2,C_3<0$. Thus for nonnegativity of the $\lambda_i$ we need 
\begin{gather}
\max\left(0, -\frac{D_1}{C_1}\right)     \le  \lambda_4 \le  \min\left( -\frac{D_2}{C_2}  , -\frac{D_3}{C_3} \right)\label{eq:ineq}
\end{gather}
The explicit expressions for $-\frac{D_1}{C_1} , -\frac{D_2}{C_2}  , -\frac{D_3}{C_3} $ are
\begin{align*}
-\frac{D_1}{C_1} ={}& 
\big((k-1)lp_1p_2 +klp_2^2 \\
{}&+k(l-1)p_2p_3 -2(l+k)p_1p_3 \big)\\
{}&\times\frac{1}{ \left( p_2+2\,p_3 \right)  \left( p_2+2 \,p_1 \right) kl}  \\[0.5ex]
-\frac{D_2}{C_2} ={}& 
\frac {(k-1)lp_1 +klp_2 +k(l-1)p_3 }
{ \left(p_2+2p_3 \right) kl} \\[0.5ex]
-\frac{D_3}{C_3} ={}& 
\frac {(k-1)lp_1 + klp_2 +k(l-1)p_3 }
{ \left(p_2+2p_1 \right) kl}
\end{align*}
Evidently $-\frac{D_2}{C_2},-\frac{D_3}{C_3}>0$.  
To see the inequalities \eqref{eq:ineq} can be satisfied 
we simply observe that 
\begin{align*}
-\frac{D_2}{C_2} - {}& \left( -\frac{D_1}{C_1}\right) =\\ 
{}&\frac {2p_1\, \left( (k-1)p_1+kp_2+(k+1)p_3 \right) }
{k \left( p_2+2p_3 \right)  \left( p_2+2p_1 \right) }
 > 0  \\[0.5ex]
-\frac{D_3}{C_3} -  {}&\left( -\frac{D_1}{C_1}\right) = \\
{}&\frac {2p_3\, \left( (l+1)p_2+lp_2+(l-1)p_3\right) }
{l \left( p_2+2p_3 \right)  \left( p_2+2p_1 \right) }
 > 0
\end{align*}
Thus both quantities on the rightmost side of \eqref{eq:ineq} are greater than 
both quantities on the leftmost side and thus the inequalities \eqref{eq:ineq} can  
always be satisfied.  Note that the quantity $-\frac{D_1}{C_1}$ is of indeterminate sign. 
When it is negative it is possible to take $\lambda_4=0$ and thus write 
$S_A B_{i,i+1}$ as a convex combination of just three extreme points. 
However, in general, four are needed.  
\end{proof}

\begin{theorem}
The closure of the convex hull of the local state space has $2^{n-1}$ extreme points in bijection with the power set of 
$\{1,2,\ldots,n-1\}$. 
\end{theorem}

\begin{proof}
By \textbf{Lemma 1}, the $S_A$ are all contained in the closure of the local state space and thus in 
the closure of its convex
hull. Thus the convex hull of the points $S_A$ is contained in the closure of the convex hull 
of the local state space. 
By \textbf{Lemma 2}, no point $S_A$ can be expressed as a nontrivial convex combination of the 
others. Thus the convex hull of the points $S_A$ has all the $2^{n-1}$ points $S_A$ as extreme points. 
By \textbf{Lemma 3}, any point in the local state space can be expressed as a convex combination of the $S_A$.
Thus the local state space is contained in the convex hull of the points $S_A$, which is a closed convex set, so 
must therefore also contain the closure of the convex hull of the local state space. 
Combining these results we conclude that the closure of the convex hull of the local state space can be identified
with the convex hull of the points $S_A$ and these are a complete set of extreme points. 
\end{proof}

\end{document}